\documentclass[12pt]{amsart}
\usepackage{bbm}
\usepackage{}
\usepackage{amsmath}
\usepackage{amsfonts}
\usepackage{amssymb}
\usepackage[all,cmtip]{xy}           

\usepackage{bbding}
\usepackage[shortlabels]{enumitem}
\usepackage{ifpdf}
\ifpdf
\usepackage[colorlinks,final,backref=page,hyperindex]{hyperref}
\else
\usepackage[colorlinks,final,backref=page,hyperindex,hypertex]{hyperref}
\fi
\usepackage{tikz}
\usepackage[active]{srcltx}

\makeatletter

\newtheorem{theorem}{Theorem}[section]
\newtheorem{prop}[theorem]{Proposition}
\newtheorem{lemma}[theorem]{Lemma}
\newtheorem{coro}[theorem]{Corollary}
\newtheorem{prop-def}{Proposition-Definition}[section]

\theoremstyle{definition}
\newtheorem{defn}[theorem]{Definition}

\newtheorem{remark}[theorem]{Remark}
\newtheorem{exam}[theorem]{Example}

\newcommand{\nc}{\newcommand}

\newcommand {\emptycomment}[1]{}

\nc{\delete}[1]{{}}
\nc{\mmargin}[1]{}

\nc{\mlabel}[1]{\label{#1}}  
\nc{\mcite}[1]{\cite{#1}}  
\nc{\mref}[1]{\ref{#1}}  
\nc{\meqref}[1]{\eqref{#1}}  
\nc{\mbibitem}[1]{\bibitem{#1}} 

\delete{
	\nc{\mlabel}[1]{\label{#1}  
		{\hfill \hspace{1cm}{\bf{{\ }\hfill(#1)}}}}
	\nc{\mcite}[1]{\cite{#1}{{\bf{{\ }(#1)}}}}  
	\nc{\mref}[1]{\ref{#1}{{\bf{{\ }(#1)}}}}  
	\nc{\meqref}[1]{\eqref{#1}{{\bf{{\ }(#1)}}}}  
	\nc{\mbibitem}[1]{\bibitem[\bf #1]{#1}} 
}

\newcommand{\g}{\mathfrak g}

 \font\cyrs=wncyr7

\nc{\vep}{\varepsilon}
\nc{\bin}[2]{ (_{\stackrel{\scs{#1}}{\scs{#2}}})}  
\nc{\binc}[2]{(\!\! \begin{array}{c} \scs{#1}\\
		\scs{#2} \end{array}\!\!)}  
\nc{\bincc}[2]{  ( {\scs{#1} \atop
		\vspace{-1cm}\scs{#2}} )}  
\nc{\oline}[1]{\overline{#1}}
\nc{\mapm}[1]{\lfloor\!|{#1}|\!\rfloor}
\nc{\bs}{\bar{S}}
\nc{\cast}{{\,\mbox{\raisebox{.8pt}{$\scriptstyle \circledast$}}\,}}
\nc{\la}{\longrightarrow}
\nc{\hot}{\widehat{\otimes}}
\nc{\ot}{\otimes}
\nc{\rar}{\rightarrow}
\nc{\dar}{\downarrow}
\nc{\dap}[1]{\downarrow \rlap{$\scriptstyle{#1}$}}
\nc{\defeq}{\stackrel{\rm def}{=}}
\nc{\dis}[1]{\displaystyle{#1}}
\nc{\dotcup}{\ \displaystyle{\bigcup^\bullet}\ }
\nc{\hcm}{\ \hat{,}\ }
\nc{\hts}{\hat{\otimes}}
\nc{\hcirc}{\hat{\circ}}
\nc{\lleft}{[}
\nc{\lright}{]}
\nc{\curlyl}{\left \{ \begin{array}{c} {} \\ {} \end{array}
	\right .  \!\!\!\!\!\!\!}
\nc{\curlyr}{ \!\!\!\!\!\!\!
	\left . \begin{array}{c} {} \\ {} \end{array}
	\right \} }
\nc{\longmid}{\left | \begin{array}{c} {} \\ {} \end{array}
	\right . \!\!\!\!\!\!\!}
\nc{\ora}[1]{\stackrel{#1}{\rar}}
\nc{\ola}[1]{\stackrel{#1}{\la}}
\nc{\scs}[1]{\scriptstyle{#1}} \nc{\mrm}[1]{{\rm #1}}
\nc{\dirlim}{\displaystyle{\lim_{\longrightarrow}}\,}
\nc{\invlim}{\displaystyle{\lim_{\longleftarrow}}\,}
\nc{\dislim}[1]{\displaystyle{\lim_{#1}}} \nc{\colim}{\mrm{colim}}
\nc{\mvp}{\vspace{0.3cm}} \nc{\tk}{^{(k)}} \nc{\tp}{^\prime}
\nc{\ttp}{^{\prime\prime}} \nc{\svp}{\vspace{2cm}}
\nc{\vp}{\vspace{8cm}}
\nc{\modg}[1]{\!<\!\!{#1}\!\!>}
\nc{\intg}[1]{F_C(#1)}
\nc{\lmodg}{\!<\!\!}
\nc{\rmodg}{\!\!>\!}
\nc{\cpi}{\widehat{\Pi}}
\nc{\ssha}{{\mbox{\cyrs X}}} 
\nc{\tsha}{{\mbox{\cyrt X}}}
\nc{\shpr}{\diamond}    
\nc{\labs}{\mid\!}
\nc{\rabs}{\!\mid}
\nc{\btr}{\blacktriangleright}

\nc{\ad}{\mrm{ad}}
\nc{\rRB}{\mathsf{rRB}}
\nc{\cocrRB}{\mathsf{cocrRB}}
\nc{\PH}{\mathsf{PH}}
\nc{\cocPH}{\mathsf{cocPH}}
\nc{\ann}{\mrm{ann}}
\nc{\Aut}{\mrm{Aut}}
\nc{\Der}{\mrm{Der}}
\nc{\Sym}{\mrm{Sym}}
\nc{\br}{\mrm{bre}}
\nc{\can}{\mrm{can}}
\nc{\Cont}{\mrm{Cont}}
\nc{\rchar}{\mrm{char}}
\nc{\cok}{\mrm{coker}}
\nc{\de}{\mrm{dep}}
\nc{\dtf}{{R-{\rm tf}}}
\nc{\dtor}{{R-{\rm tor}}}

\nc{\Dif}{\mrm{Diff}}
\nc{\Div}{\mrm{Div}}
\nc{\End}{\mrm{End}}
\nc{\Ext}{\mrm{Ext}}
\nc{\Fil}{\mrm{Fil}}
\nc{\Fr}{\mrm{Fr}}
\nc{\Frob}{\mrm{Frob}}
\nc{\Gal}{\mrm{Gal}}
\nc{\GL}{\mrm{GL}}
\nc{\Gr}{\mrm{Gr}}
\nc{\Hom}{\mrm{Hom}}
\nc{\Hoch}{\mrm{Hoch}}
\nc{\hsr}{\mrm{H}}
\nc{\hpol}{\mrm{HP}}
\nc{\id}{\mrm{id}}
\nc{\im}{\mrm{im}}
\nc{\inv}{\mrm{inv}}
\nc{\Id}{\mrm{Id}}
\nc{\ID}{\mrm{ID}}
\nc{\Irr}{\mrm{Irr}}
\nc{\incl}{\mrm{incl}}
\nc{\length}{\mrm{length}}
\nc{\NLSW}{\mrm{NLSW}}
\nc{\Lie}{\mrm{Lie}}
\nc{\mchar}{\rm char}
\nc{\mpart}{\mrm{part}}
\nc{\ql}{{\QQ_\ell}}
\nc{\qp}{{\QQ_p}}
\nc{\rank}{\mrm{rank}}
\nc{\rcot}{\mrm{cot}}
\nc{\rdef}{\mrm{def}}
\nc{\rdiv}{{\rm div}}
\nc{\rtf}{{\rm tf}}
\nc{\rtor}{{\rm tor}}
\nc{\res}{\mrm{res}}
\nc{\SL}{\mrm{SL}}
\nc{\Spec}{\mrm{Spec}}
\nc{\tor}{\mrm{tor}}
\nc{\Tr}{\mrm{Tr}}
\nc{\tr}{\mrm{tr}}
\nc{\wt}{\mrm{wt}}

\nc{\bfk}{{\bf k}}
\nc{\bfone}{{\bf 1}}
\nc{\bfzero}{{\bf 0}}
\nc{\detail}{\marginpar{\bf More detail}
	\noindent{\bf Need more detail!}
	\svp}
\nc{\gap}{\marginpar{\bf Incomplete}\noindent{\bf Incomplete!!}
	\svp}
\nc{\FMod}{\mathbf{FMod}}
\nc{\Int}{\mathbf{Int}}
\nc{\Mon}{\mathbf{Mon}}
\nc{\remarks}{\noindent{\bf Remarks: }}
\nc{\Rep}{\mathbf{Rep}}
\nc{\Rings}{\mathbf{Rings}}
\nc{\Sets}{\mathbf{Sets}}
\nc{\Diff}{\mathbf{Diff}}
\nc{\Inte}{\mathbf{Inte}}
\nc{\U}{\mathbf{U}}

\nc{\BA}{{\mathbb A}}   \nc{\CC}{{\mathbb C}}
\nc{\DD}{{\mathbb D}}   \nc{\EE}{{\mathbb E}}
\nc{\FF}{{\mathbb F}}   \nc{\GG}{{\mathbb G}}
\nc{\HH}{{\mathbb H}}   \nc{\LL}{{\mathbb L}}
\nc{\NN}{{\mathbb N}}   \nc{\PP}{{\mathbb P}}
\nc{\QQ}{{\mathbb Q}}   \nc{\RR}{{\mathbb R}}
\nc{\TT}{{\mathbb T}}   \nc{\VV}{{\mathbb V}}
\nc{\ZZ}{{\mathbb Z}}   \nc{\TP}{\widetilde{P}}
\newcommand{\huaO}{{\mathcal{O}}}

\nc{\cala}{{\mathcal A}}    \nc{\calc}{{\mathcal C}}
\nc{\calb}{{\mathcal B}}
\nc{\cald}{\mathcal{D}}     \nc{\cale}{{\mathcal E}}
\nc{\calf}{{\mathcal F}}    \nc{\calg}{{\mathcal G}}
\nc{\calh}{{\mathcal H}}    \nc{\cali}{{\mathcal I}}
\nc{\call}{{\mathcal L}}    \nc{\calm}{{\mathcal M}}
\nc{\caln}{{\mathcal N}}    \nc{\calo}{{\mathcal O}}
\nc{\calp}{{\mathcal P}}    \nc{\calr}{{\mathcal R}}

\nc{\cals}{{\mathcal S}}    \nc{\calt}{{\mathcal T}}
\nc{\calv}{{\mathcal V}}    \nc{\calw}{{\mathcal W}}
\nc{\calx}{{\mathcal X}}

\nc{\fraka}{{\mathfrak a}}
\nc{\frakb}{\mathfrak{b}}
\nc{\frake}{{\frak e}}
\nc{\frakg}{{\frak g}}
\nc{\frakh}{{\frak h}}
\nc{\fraki}{{\frak i}}
\nc{\frakl}{{\frak l}}
\nc{\fraks}{{\frak s}}
\nc{\frakB}{{\frak B}}
\nc{\frakm}{{\frak m}}
\nc{\frakM}{{\frak M}}
\nc{\frakp}{{\frak p}}
\nc{\frakW}{{\frak W}}
\nc{\frakX}{{\frak X}}
\nc{\frakS}{{\frak S}}
\nc{\frakA}{{\frak A}}
\nc{\frakG}{{\frak G}}
\nc{\frakP}{{\frak P}}
\nc{\frakx}{{\frakx}}

\nc{\ynr}[1]{\textcolor{orange}{\underline{Yunnan:}#1 }}

\nc{\lir}[1]{\textcolor{red}{\underline{Li:}#1 }}

\delete{
	\input cyracc.def

	\newtheorem{theorem}{Theorem}[section]

	\theoremstyle{definition}

	\theoremstyle{remark}
	\newtheorem{remark}[theorem]{Remark}

	\allowdisplaybreaks
	
	\numberwithin{equation}{section}
	
}

\begin{document}

\title[Rota-Baxter operators of weight 0 on groups]{Relative Rota-Baxter operators of weight 0 on groups, pre-groups, braces, the Yang-Baxter equation and $T$-structures}

\author{Yunnan Li}
\address{School of Mathematics and Information Science, Guangzhou University,
Guangzhou 510006, China}
\email{ynli@gzhu.edu.cn}

\author{Yunhe Sheng}
\address{Department of Mathematics, Jilin University, Changchun 130012, Jilin, China}
\email{shengyh@jlu.edu.cn}

\author{Rong Tang}
\address{Department of Mathematics, Jilin University, Changchun 130012, Jilin, China}
\email{tangrong@jlu.edu.cn}


\begin{abstract}
In this paper, we study relative Rota-Baxter operators of weight $0$ on groups and give various examples. In particular, we propose different approaches to study Rota-Baxter operators of weight $0$ on groups and Lie groups. We establish various explicit relations among relative Rota-Baxter operators of weight $0$ on groups, pre-groups, braces, set-theoretic solutions
of the Yang-Baxter equation and $T$-structures.
\end{abstract}

\keywords{relative Rota-Baxter operator, pre-group, brace, the Yang-Baxter equation, $T$-structure\\
\qquad 2020 Mathematics Subject Classification. 16T25,   17B38}

\maketitle

\tableofcontents

\allowdisplaybreaks

\section{Introduction}

The concept of Rota-Baxter operators on associative algebras was
introduced in 1960 by G. Baxter \cite{Ba} in his study of
fluctuation theory in probability. Recently it has found many
applications, including in Connes-Kreimer's~\cite{CK} algebraic
approach to the renormalization in perturbative quantum field
theory. In the Lie algebra context, a Rota-Baxter operator of
weight $0$ was introduced independently in the 1980s as the
operator form of the classical Yang-Baxter equation \cite{STS}. For further details on
Rota-Baxter operators, see~\cite{Gub}.
The more general notion of a relative Rota-Baxter operator of weight $0$ (originally  called
$\huaO$-operator)
on a Lie algebra was introduced by Kupershmidt~\cite{Ku} to better understand the classical Yang-Baxter equation. This structure can be traced back to Bordemann's work on integrable systems \cite{Bor}. A  Rota-Baxter operator of weight 0 naturally gives rise to pre-Lie algebras \cite{Aguiar}. The notion of Rota-Baxter operators of weight $1$ was introduced by Bai, Guo and Ni with applications in Lax pairs and post-Lie algebras \cite{BGN}.

In \cite{GLS}, the notion of Rota-Baxter operators of weight $\pm1$ on groups was introduced, with applications to factorizations of groups. The differentiation of a  Rota-Baxter operator  of  weight $1$ on a Lie group yields a Rota-Baxter operator  of  weight $ 1$ on the associated Lie algebra. Motivated by this structure, Goncharov introduced the notion of a Rota-Baxter operator on a Hopf algebra, and open a new research area~\cite{Go}. It was observed by Bardakov and Gubarev in \cite{BG} that  Rota-Baxter operators of weight $1$ on groups induces skew left braces and solutions of the Yang-Baxter equation. See \cite{CS} for more studies on skew left braces and Rota-Baxter operators. The more general notion of a relative Rota-Baxter operator on a group with respect to an action on another group was introduced in \cite{JSZ}, and further characterized by matched pairs of groups with applications in post-groups and the Yang-Baxter equation \cite{BGST}.

Note that in the Lie algebra context, (relative) Rota-Baxter operators of weight 0 are of great interest, while in the group context, recent studies on (relative) Rota-Baxter operators on groups mainly focus on the weight $1$ case. Considering the category of relative Rota-Baxter operators on groups, it is very easy to obtain the correct notion of a relative Rota-Baxter operator of weight $0$ by requiring the source group to be abelian. Nevertheless, this approach can not be applied to define what is a Rota-Baxter operator of weight $0$ on a group. So how to define a Rota-Baxter operator of weight $0$ is still an open question.
The first purpose of this paper is to solve this question. We propose two different approaches, one for arbitrary groups and one for Lie groups. More precisely, we introduce Rota-Baxter operators of weight $0$ on an arbitrary group $G$ using the algebraic construction of the Malcev completion of a group due to Quillen~\cite{Qu}. For a Lie group $G$, a Rota-Baxter operator of weight $0$ is simply defined to be a relative Rota-Baxter operator from $\g$ to $G$ with respect to the adjoint action, where $\g$ is the Lie algebra of $G$. The relation between these two approaches is still mysterious to us.

To construct solutions of the Yang-Baxter equation, Rump introduced  the notion of braces in~\cite{Ru}, which is   a generalization of radical rings. Further studies were carried out in~\cite{CJO,G,Sm18}.
Recently, braces were generalized to skew left braces by Guarnieri and Vendramin in \cite{GV} to construct
non-degenerate and not necessarily involutive  solutions of the Yang-Baxter equation. 
As shown in \cite[Proposition 3.1]{BG}, any Rota-Baxter group $G$ of weight $1$ naturally produces a skew left brace structure on $G$. When is it especially a brace? A naive answer is given by requiring $G$ to be an abelian group. But the situation becomes frustrating. As if $G$ is abelian, any Rota-Baxter operator of weight $1$ on $G$ degenerates to be an abelian group endomorphism, and we only get the trivial brace.
In order to obtain nontrivial braces through this approach, we need to consider (relative) Rota-Baxter operators of weight $0$ on groups. This is another motivation for us to study the weight 0 case.
Note  that Bardakov and Nikonov introduced Rota-Baxter operators of weight $0$ on Lie groups from another perspective quite recently in  ~\cite{BN}.

As aforementioned, braces can be used to construct  solutions of the Yang-Baxter equation, so relative Rota-Baxter operators of weight 0 can also give rise to solutions of the Yang-Baxter equation. On the other hand, a non-degenerate involutive set-theoretic solution of the
Yang-Baxter equation also naturally gives rise to a relative Rota-Baxter operator of weight 0.

The paper is organized as follows. In Section \ref{sec:rRB}, we revisit relative Rota-Baxter operators on groups and give various examples. In Section \ref{sec:RB}, we propose two different approaches to define Rota-Baxter operators of weight 0 on groups and Lie groups. In Section \ref{sec:pre}, we establish the relation between relative Rota-Baxter operators of weight 0 and pre-groups. In Section \ref{sec:bra}, we show that a relative Rota-Baxter of weight 0 induces a brace, and we also answer the question when a brace is induced from a relative Rota-Baxter operator of weight 0. In Section \ref{sec:YBE}, we construct set-theoretic solutions of the Yang-Baxter equation from a relative Rota-Baxter operator of weight 0.  In Section \ref{sec:T}, we show that a relative Rota-Baxter operator of weight 0 naturally gives rise to a $T$-structure.

\section{Relative Rota-Baxter operators of weight 0 on groups and examples}\label{sec:rRB}
In this section, we revisit relative Rota-Baxter operator of weight 0 and give various examples.

\begin{defn}
Given a group $G$ and an abelian group $(V,+)$ with a group homomorphism $\Phi:G\to \Aut(V)$,
namely $(V,\Phi)$ is a $\mathbb Z G$-module, a map $\calr:V\to G$ is called a {\bf relative Rota-Baxter operator of weight 0} on $G$ with respect to $(V,\Phi)$, if
\begin{equation}\label{eq:RRB0}
\calr(u)\calr(v)=\calr(u+\Phi(\calr(u))v),\quad \forall u,v\in V.
\end{equation}
\end{defn}

Eq.~\eqref{eq:RRB0} is equivalent to the following condition:
\begin{equation}\label{eq:RRB0'}
\calr(u)\calr(\Phi(\calr(u)^{-1})v)=\calr(u+v),\quad \forall u,v\in V.
\end{equation}
In particular, we have
\begin{eqnarray}
\label{eq:rrb-unit}
\calr(0)&=&e_G,\\
\label{eq:rrb-minus}
\calr(-u) &=& \calr\left(\Phi(\calr(u)^{-1})u\right)^{-1},\quad\forall u\in V,\\
\label{eq:rrb-inverse}
\calr(u)^{-1} &=& \calr\left(-\Phi(\calr(u)^{-1})u\right),\quad\forall u\in V.
\end{eqnarray}
Also, the commutativity of $(V,+)$ implies the following commutation relation:
$$\calr(u)\calr(\Phi(\calr(u)^{-1})v)=\calr(v)\calr(\Phi(\calr(v)^{-1})u),\quad \forall u,v\in V.$$

\begin{prop}\label{prop:rrb-graph}
Given a group $G$ and a $\mathbb Z G$-module $(V,\Phi)$, a map $\calr:V\to G$ is a relative Rota-Baxter operator if and only if the graph of $\calr$
$${\rm Graph}(\calr)=\{(u,\calr(u))\,|\,u\in V\}$$
is a subgroup of the semidirect product $V\rtimes_\Phi G$.
\end{prop}
\begin{proof}
For all $u,v\in V$, we have
$$
(u,\calr(u))(v,\calr(v)) = (u+\Phi(\calr(u))v,\calr(u)\calr(v)),
$$
which implies that ${\rm Graph}(\calr)$ is subgroup if and only if Eq.~\eqref{eq:RRB0} holds for $\calr$.
\end{proof}


\begin{exam}\label{ex:Z}
Obviously,  any $G$-module structure on $\mathbb Z$ is given by a group homomorphism $\chi:G\to \mathbb Z^*=\{\pm1\}$. Then a map $t:\mathbb Z\to G$ is a relative Rota-Baxter operator of weight $0$ on $G$ with respect to $\chi$ if and only if
$$t_mt_n=t_{m+\chi(t_m)n},\quad \forall m,n\in\mathbb Z. $$

In particular, if we also take $G=\mathbb Z$, there are only two group homomorphisms $\chi_1,\, \chi_2$ from $\mathbb Z$ to $\{\pm1\}$, namely
$$\chi_1(n)=1,\quad \chi_2(n)=(-1)^n,\quad \forall n\in \mathbb Z.$$
Then any relative Rota-Baxter operator $t:\mathbb Z\to \mathbb Z$
is easily checked to be
a group endomorphism of $\mathbb Z$.
\end{exam}


Next we give a simple class of non-bijective relative Rota-Baxter operators.

\begin{exam}\label{ex:semi-direct}
Given an abelian group $(G,\cdot)$ and a $\mathbb Z G$-module $(A,+,\rho)$, let $V$ be the direct product $A\times G$ with the following group action $\Phi:G\to\Aut(V)$,
$$\Phi(x)(a,y)=(\rho(x)a,y),\quad\forall a\in A,\,x,y\in G.$$
Then the natural projection $\calr:V\to G,\ (a,x)\mapsto x$ is clearly a relative Rota-Baxter operator.
\end{exam}

In the sequel, we study two kinds of relative Rota-Baxter operators of weight 0 on the symmetric group $S_n$ with respect to its sign representation and its permutation representation respectively.
\begin{exam}
Let $(\mathbb Z,\Phi)$ be the sign representation of $\mathbb Z S_n$, that is,
$$\Phi(\sigma)i=(-1)^{\ell(\sigma)}i,\quad \forall \sigma\in S_n,\,i\in\mathbb Z.$$
According to Eq.~\eqref{eq:RRB0'}, a map $\mathbb Z\to S_n$ is a relative Rota-Baxter operator if and only if
$$\calr(i)\calr((-1)^{\ell(\calr(i))}j)=\calr\left(i+j\right)=\calr(j)\calr((-1)^{\ell(\calr(j))}i)
,\quad \forall i,j\in\mathbb Z.$$

(a) If $\calr(1)$ is a permutation of odd length, then
 $\calr(\pm1)$ are involutions in $S_n$, as
$$\calr(1)\calr\left((-1)^{\ell(\calr(1))+1}\right)=\calr(-1)\calr\left((-1)^{\ell(\calr(-1))}\right)=\calr(0)=1$$
by taking $i=1$ and $j=-1$ in the above condition. More generally, since
$$\calr(1)\calr(-k)=\calr(k+1)=\calr(k)\calr\left((-1)^{\ell(\calr(k))}\right),\quad \forall k\in\mathbb Z,$$
we know that the parity of $\ell(\calr(k))$ is the same as that of $k$. Hence,
$$\calr(1)\calr(-k)=\calr(k+1)=\calr(k)\calr\left((-1)^k\right),\quad \forall k\in\mathbb Z.$$
Similarly, we have
$$\calr(-1)\calr(k)=\calr(-k-1)=\calr(-k)\calr\left((-1)^{-k-1}\right),\quad \forall k\in\mathbb Z.$$
Therefore, we obtain by induction that
$$\calr(k)=
\begin{cases}
\calr(1)\calr(-1)\cdots \calr((-1)^{k-1}),& k>0,\\
1,&n=0,\\
\calr(-1)\calr(1)\cdots \calr((-1)^{-k}),& k<0.
\end{cases}$$
Any choice of involutions for $\calr(\pm 1)$ together with such a formula of $\calr(k)$ define a relative Rota-Baxter operator $\calr:\mathbb Z\to S_n$.

(b) If $\calr(1)$ is a permutation of even length, then
$\calr(k)=\calr(1)^k$ for any $k\in\mathbb Z$. Namely, $\calr$ becomes a group homomorphism from $(\mathbb Z,+)$ to $S_n$ factoring through the alternating group $A_n$.

\end{exam}

\begin{theorem}\label{thm:perm_repn}
Let $V_n=(\mathbb Z^n,\Phi)$ be the permutation representation of $S_n$ with the standard basis $\{e_i\}_{1\leq i\leq n}$, namely
$$\Phi(w)e_i=e_{w(i)},\quad \forall w\in S_n,$$
then any two tuples $\sigma=(\sigma_i)_{1\leq i\leq n}$ and $\bar\sigma=(\bar\sigma_i)_{1\leq i\leq n}$ of permutations in $S_n$ uniquely determine a relative Rota-Baxter operator $\calr_{\sigma,\,\bar\sigma}:V_n\to S_n$ such that $\calr_{\sigma,\,\bar\sigma}(e_i)=\sigma_i$ and $\calr_{\sigma,\,\bar\sigma}(-e_i)=\bar\sigma_i$ for $1\leq i\leq n$, if they satisfy
\begin{equation}\label{eq:perm_repn}
\begin{cases}
\sigma_i\bar\sigma_{\sigma_i^{-1}(i)}=\bar\sigma_i\sigma_{\bar\sigma_i^{-1}(i)}=1,& 1\leq i\leq n,\\
\sigma_j\sigma_{\sigma_j^{-1}(k)}=\sigma_k\sigma_{\sigma_k^{-1}(j)},& 1\leq j<k\leq n,\\
\bar\sigma_j\bar\sigma_{\bar\sigma_j^{-1}(k)}=\bar\sigma_k\bar\sigma_{\bar\sigma_k^{-1}(j)},& 1\leq j<k\leq n,\\
\sigma_j\bar\sigma_{\sigma_j^{-1}(k)}=\bar\sigma_k\sigma_{\bar\sigma_k^{-1}(j)},& 1\leq j,k\leq n.
\end{cases}
\end{equation}

In particular, any tuple $\sigma=(\sigma_i)_{1\leq i\leq n}$ of permutations in $S_n$ satisfying
\begin{equation}\label{eq:perm_repn'}
\begin{cases}
\sigma_i\sigma_{\sigma_i^{-1}(i)}=1,& 1\leq i\leq n,\\
\sigma_j\sigma_{\sigma_j^{-1}(k)}=\sigma_k\sigma_{\sigma_k^{-1}(j)},& 1\leq j<k\leq n,
\end{cases}
\end{equation}
defines a relative Rota-Baxter operator $\calr_\sigma=\calr_{\sigma,\,\sigma}:V_n\to S_n$.
\end{theorem}
\begin{proof}
According to Eq.~\eqref{eq:RRB0'}, any relative Rota-Baxter operator $\calr:V_n\to S_n$ satisfies the stated conditions, when we take $\sigma_i=\calr(e_i)$ and $\bar\sigma_i=\calr(-e_i)$ for $1\leq i\leq n$. Conversely, we prove that these two specific tuples $(\sigma_i)_{1\leq i\leq n}$ and $(\bar\sigma_i)_{1\leq i\leq n}$ are sufficient to determine a relative Rota-Baxter operator. We only need to show the existence of such an operator, and Eq.~\eqref{eq:RRB0'} guarantees its uniqueness.

For any $v=\sum_i a_i e_i\in V_n$, let $|v|=\sum_i |a_i|$.
Now define a map $\calr:V_n\to S_n$ as follow. Set $\calr(0)=1$, $\calr(e_i)=\sigma_i$ and $\calr(-e_i)=\bar\sigma_i$ for $1\leq i\leq n$. When $v\neq0$, there exists $1\leq i\leq n$ such that $v=v'\pm e_i$ and $|v|=|v'|+1$, and we recursively define
$$\calr(v)
=\begin{cases}
\sigma_i\calr(\sigma_i^{-1}v'),&\mbox{if }v=v'+ e_i,\\
\bar\sigma_i\calr(\bar\sigma_i^{-1}v'),&\mbox{if }v=v'- e_i,
\end{cases}$$
by induction on $|v|$, as $|\sigma_i^{-1}v'|=|\bar\sigma_i^{-1}v'|=|v'|<|v|$. Such a recursive definition of $\calr$ is independent of the choice of $i$. Indeed, if $v=v'+e_i=v''+e_j=w+e_i+e_j\ (i\neq j)$ and $|v|=|v'|+1=|v''|+1$ (Other cases are similar to check), then
\begin{align*}
\calr(v)=\sigma_i\calr(\sigma_i^{-1}v')
=\sigma_i\calr\left(\sigma_i^{-1}w+e_{\sigma_i^{-1}(j)}\right)
=\sigma_i\sigma_{\sigma_i^{-1}(j)}\calr\left(\sigma_{\sigma_i^{-1}(j)}^{-1}\sigma_i^{-1}w\right).
\end{align*}
By symmetry, we also have
\begin{align*}
\calr(v)
&=\sigma_j\sigma_{\sigma_j^{-1}(i)}\calr\left(\sigma_{\sigma_j^{-1}(i)}^{-1}\sigma_j^{-1}w\right),
\end{align*}
and two expressions give the same element by the compatibility condition~\eqref{eq:perm_repn}.

Next we show that the so-defined operator $\calr:V_n\to G$ is
a relative Rota-Baxter operator satisfying Eq.~\eqref{eq:RRB0'} by induction on $|u|$. When $|u|=0$, namely $u=0$, it is clear.
Otherwise, let $u=u'+ e_i$ and $|u|=|u'|+1$ without loss of generality. Then
\begin{align*}
\calr(u)\calr(\calr(u)^{-1}v)
&=\sigma_i\calr(\sigma_i^{-1}u')\calr\left((\sigma_i\calr(\sigma_i^{-1}u'))^{-1}v\right)\\
&=\sigma_i\calr(\sigma_i^{-1}u')\calr\left(\calr(\sigma_i^{-1}u')^{-1}\sigma_i^{-1}v\right)\\
&=\sigma_i\calr(\sigma_i^{-1}u'+\sigma_i^{-1}v)\\
&=\sigma_i\calr(\sigma_i^{-1}(u'+v)),
\end{align*}
where the first equality is by the definition of $\calr$, and the third one is by the induction hypothesis as $|\sigma_i^{-1}u'|=|u'|<|u|$. If $|u'+v|<|u+v|$, then
$$\calr(u)\calr(\calr(u)^{-1}v)=\sigma_i\calr(\sigma_i^{-1}(u'+v))=\calr((u'+v)+e_i)=\calr(u+v).$$
Otherwise, $u'+v=(u+v)-e_i$ and $|u'+v|=|u+v|+1$, so
\begin{align*}
\calr(u)\calr(\calr(u)^{-1}v)
&=\sigma_i\calr(\sigma_i^{-1}(u'+v))\\
&=\sigma_i\calr(\sigma_i^{-1}(u+v)-e_{\sigma_i^{-1}(i)})\\
&=\sigma_i\bar\sigma_{\sigma_i^{-1}(i)}\calr(\bar\sigma_{\sigma_i^{-1}(i)}^{-1}\sigma_i^{-1}(u+v))\\
&=\calr(u+v),
\end{align*}
where the third equality is by the definition of $\calr$, and the last one is due to the compatibility condition~\eqref{eq:perm_repn}.
\end{proof}

\begin{exam}
Let $V_3$ be the permutation module of $S_3$.
According to Theorem~\ref{thm:perm_repn}, we can figure out all relative Rota-Baxter operators of $S_3$ with respect to the permutation representation on $V_3$. Indeed, any pair of distinct tuples $\{\sigma,\bar\sigma\}$ satisfying Eq.~\eqref{eq:perm_repn} is included in one of the following two cases:
$$
\{((123),(123),(123)),\ ((132),(132),(132))\},\quad
\{((132),(132),(132)),\ ((123),(123),(123))\},
$$
while any tuple $\sigma$ satisfying Eq.~\eqref{eq:perm_repn'} belongs to one of the following 10 cases:
$$\begin{array}{lllll}
((1),(1),(1)), & ((1),(1),(12)), & ((1),(23),(23)), & ((1),(13),(1)), & ((23),(1),(1)),\\[.5em]
((23),(23),(23)), & ((12),(12),(1)), & ((12),(12),(12)), & ((13),(1),(13)), & ((13),(13),(13)).
\end{array}$$

For the case when $n=4$, there are already 88 tuples $\sigma$ satisfying Eq.~\eqref{eq:perm_repn'}. For example, we can take
$$\sigma=((24),(13),(1432),(1234)),$$
and these 4 permutations in $\sigma$ generate a subgroup of $S_4$ isomorphic to the dihedral group $D_4$ of order 8.
\end{exam}


\section{Rota-Baxter operators of weight 0 on (Lie) groups}\label{sec:RB}

In this section, we introduce Rota-Baxter operators of weight 0 on groups. For Lie groups, the approach is more adapted; while for arbitrary groups, the approach is more fancy.

\subsection{Rota-Baxter operators of weight 0 on  groups}

We need some preparation to introduce Rota-Baxter operators of weight $0$ on an arbitrary group $G$, based on the algebraic construction of the Malcev completion of a group due to Quillen~\cite{Qu}. See also the survey~\cite{Me}.


Let $\mathbbm k$ be an algebraically closed field, and  of characteristic 0. First consider the group ring $\mathbbm{k} G$ with the standard Hopf algebra structure  $(\mathbbm{k} G,\cdot,\Delta,\vep,S)$ such that
$$\Delta(x)=x\otimes x,\quad \vep(x)=1,\quad S(x)=x^{-1},\quad\forall x\in G.$$
Let
$$I=\ker\vep={\rm span}_\mathbbm k\{x-e_G\,|\,x\in G\}$$
be the augmentation ideal of $\mathbbm{k} G$.
Denote
$$\widehat{\mathbbm k G}=\invlim \mathbbm k G/I^n
=\left\{x=\sum_{i=0}^\infty x_i\,\bigg|\,\sum_{i=0}^n x_i\in \mathbbm k G/I^n,\ n\geq0\right\}. $$

Since the coproduct $\Delta$ of $\mathbbm k G$ induces a liner map
$$\Delta:\mathbbm k G/I^n\to \bigoplus_{i+j=n}\mathbbm k G/I^i\otimes \mathbbm k G/I^j,$$ and so a coassociative map $\Delta:\widehat{\mathbbm k G}\to \widehat{\mathbbm k G}\,\widehat{\otimes}\,\widehat{\mathbbm k G}$ by taking inverse limit,
the completion $\widehat{\mathbbm k G}$ is a complete Hopf algebra (with unit $\mathbbm 1$) consisting of formal power series of group ring elements; e.g. see such a classical approach in~\cite[Appendix A]{Qu}. Let $\widehat{I}$ be the completion of $I$ in $\widehat{\mathbbm k G}$. Denote $\widehat{G}$ the set of group-like elements in $\widehat{\mathbbm k G}$, and $\widehat{\frakg}$ the primitive Lie algebra of $\widehat{\mathbbm k G}$, namely
\begin{align*}
\widehat{G}&={\rm Gp}(\widehat{\mathbbm k G})=\{f\in \mathbbm 1+\widehat{I}\,|\,\Delta(f)=f\hot f\},\\
\widehat{\frakg}&={\rm Prim}(\widehat{\mathbbm k G})=\{f\in \widehat{I}\,|\,\Delta(f)=f\hot \mathbbm 1 + \mathbbm 1\hot f\}.
\end{align*}

It was proved in \cite[Theorem~3.3]{Ha} that if $H_1(G,\mathbbm k)\simeq I/I^2$ is finite-dimensional (e.g. $G$ is finitely generated), then $\widehat{G}$ is the {\bf prounipotent completion} (also called the {\bf Malcev completion}) of $G$. In fact, let
$\widehat{G}_n=\widehat{G}\cap (\mathbbm 1+\widehat{I}^n)$, then
$\widehat{G}^n=\widehat{G}/\widehat{G}_{n+1}$ is a unipotent algebraic group over $\mathbbm k$ lying in $\mathbbm 1+\widehat{I}/\widehat{I}^{n+1}$, with its Lie algebra $\widehat{\frakg}^n\subset\widehat{I}/\widehat{I}^{n+1}$, such that
$$\widehat{G}\simeq\invlim \widehat{G}^n,\quad \widehat{\frakg}\simeq\invlim \widehat{\frakg}^n.$$
As $\mathbbm k$ is of characteristic 0, the logarithm and exponential functions are well-defined and mutually inverse homeomorphisms,
$$\log:\mathbbm 1+\widehat{I}\to \widehat{I},\quad \exp:\widehat{I}\to \mathbbm 1+\widehat{I},$$
which restrict to
$$\log:\widehat{G}\to \widehat{\frakg},\quad \exp:\widehat{\frakg}\to \widehat{G}.$$
Also, $\widehat{G}$ can act on $\widehat{\frakg}$ by conjugation, so does $G$.
\begin{defn}
Given any group $G$, a map $\mathcal B:\widehat{\frakg}\to G$ is called a {\bf Rota-Baxter operator of weight $\mathbf{0}$} on $G$, if
\begin{equation}\label{eq:RRB0_comp}
\mathcal B(f)\mathcal B(g)=\mathcal B(f+\mathcal B(f)g\mathcal B(f)^{-1}),\quad \forall f,g\in\widehat{\frakg}.
\end{equation}
\end{defn}

\begin{exam}\label{ex:free_ass}
For the free associative algebra $\mathfrak{ass}_n = \mathbbm k\langle x_1,\dots,x_n\rangle$ generated by $x_1,\dots,x_n$, it is a Hopf algebra endowed with the coshuffle coproduct $\Delta$ satisfying
$$\Delta(x_i)=x_i\otimes 1 + 1\otimes x_i,\quad 1\leq i\leq n.$$
The degree completion (with degrees of the generators $x_1,\dots,x_n$ set to be 1)
$$\widehat{\mathfrak{ass}}_n = \mathbbm k\langle\langle x_1,\dots,x_n\rangle\rangle$$
of $\mathfrak{ass}_n$ is a complete Hopf algebra.
Let $\mathfrak{lie}_n$ be the Lie subalgebra of $\mathfrak{ass}_n$  generated by $x_1,\dots,x_n$ using the commutator as its Lie bracket, and $\widehat{\mathfrak{lie}}_n$ the degree completion of $\mathfrak{lie}_n$. Then $\widehat{\mathfrak{ass}}_n$ is precisely the completed universal enveloping algebra of $\widehat{\mathfrak{lie}}_n$ such that
${\rm Prim}(\widehat{\mathfrak{ass}}_n)=\widehat{\mathfrak{lie}}_n$ and
$$\frakG_n={\rm Gp}(\widehat{\mathfrak{ass}}_n)=\exp(\widehat{\mathfrak{lie}}_n)=\left\{e^f=\sum_{k=0}^\infty\dfrac{f^k}{k!}\in \widehat{\mathfrak{ass}}_n\,\bigg|\,f\in \widehat{\mathfrak{lie}}_n\right\}.$$
Moreover, the group $\frakG_n$ is the Malcev completion of the free group in $n$ letters.

Let $\calb:\widehat{\mathfrak{lie}}_n\to \frakG_n$ be a Rota-Baxter operator of weight $0$ on $\frakG_n$. By Eq.~\eqref{eq:RRB0_comp} and the Baker-Campbell-Hausdorff formula, there exists a unique  map $B:\widehat{\mathfrak{lie}}_n\to \widehat{\mathfrak{lie}}_n$ such that $\calb(x)=e^{B(x)}$ and
\begin{equation}\label{eq:B}
 {\rm BCH}(B(x),B(y))=B\left(x+e^{B(x)}ye^{-B(x)}\right)
\end{equation}
for all $x,y\in \widehat{\mathfrak{lie}}_n$, where
$${\rm BCH}(x,y)=\log(e^xe^y)=x+y+\dfrac{1}{2}[x,y]+\dfrac{1}{12}[x,[x,y]]-\dfrac{1}{12}[y,[y,x]]\,\cdots.$$
Note that ${\rm BCH}$ provides a group multiplication $*$ on $\widehat{\mathfrak{lie}}_n$ such that $(\widehat{\mathfrak{lie}}_n,*)$ acts on $(\widehat{\mathfrak{lie}}_n,+)$ by conjugation after lifting to $\frakG_n$, so Eq.~\eqref{eq:B} tells us that $B$ is also a relative Rota-Baxter operator of weight $0$.
\end{exam}

%

\subsection{Rota-Baxter operators of weight $0$ on Lie groups}

For Lie groups, there is a direct way to define Rota-Baxter operators of weight $0$.

\begin{defn}
Given a Lie group $G$ with its associated Lie algebra $\frak g$, a smooth map $\calb:\frak g\to G$ is called a {\bf Rota-Baxter operator of weight $\mathbf{0}$} on $G$, if
\begin{equation}\label{eq:rbo_lie}
\calb(x)\calb(y)=\calb(x+{\rm Ad}_{\calb(x)}y),\quad \forall x,y\in\frak g.
\end{equation}
\end{defn}


\begin{theorem}\label{coro:diff_rbo}
For a Rota-Baxter operator $\calb:\frak g\to G$ of weight $0$ on a Lie group $G$, the differentiation $B=\mathcal B|_{*0}$ of $\mathcal B$ is a Rota-Baxter operator of weight $0$ on $\frak g$, i.e. the following equality holds:
$$[B(x),B(y)]=B([B(x),y]+[x,B(y)]),\quad \forall x,y\in \frak g.$$
\end{theorem}
\begin{proof}
Since $B=\calb|_{\ast 0}$ is the tangent map of $\calb$ at $0$ by viewing $\frakg$ as an abelian Lie group, we have the following relation for sufficiently small $t\in \RR$:
$$\dfrac{d}{dt}\Big|_{t=0}\calb(tx)=\dfrac{d}{dt}\Big|_{t=0}\exp^{tB(x)}=B(x),\quad \forall x\in \frakg.$$

For any $x,y\in \frakg$,
\emptycomment{
and $a,b\in \RR^*$,
\begin{eqnarray*}
R(au+bv) &=& \dfrac{d}{dt}\Big|_{t=0}\calr(tau+tbv)\\
&\stackrel{\eqref{eq:RRB0}}{=}& \dfrac{d}{dt}\Big|_{t=0}\calr(tau)\calr\left(\Phi(\calr(tau)^{-1})tbv\right)\\
&=& \dfrac{d}{dt}\Big|_{t=0}\calr(tau) + \dfrac{d}{dt}\Big|_{t=0}\calr\left(\Phi(\calr(tau)^{-1})tbv\right)\\
&=& a\dfrac{d}{ds}\Big|_{s=0}\calr(su) + b\dfrac{d}{ds}\Big|_{s=0}\calr\left(\Phi(\calr(sab^{-1}u)^{-1})sv\right)\\
&=& aR(u) + b\calr|_{\ast 0}\left(\dfrac{d}{ds}\Big|_{s=0}\Phi(\calr(sab^{-1}u)^{-1})sv\right)\\
&=& aR(u) + bR(v),
\end{eqnarray*}
so $R:V\to \frak g$ is an $\RR$-linear map.
}
we have
\begin{eqnarray*}
[B(x),B(y)] &=& \dfrac{d^2}{dsdt}\Big|_{s,t=0}\exp^{sB(x)}\exp^{tB(y)}\exp^{-sB(x)}\\
&=& \dfrac{d^2}{dsdt}\Big|_{s,t=0} \calb(sx)\calb(ty)\calb(sx)^{-1}\\
&\stackrel{\eqref{eq:rrb-inverse}}{=}& \dfrac{d^2}{dsdt}\Big|_{s,t=0} \calb(sx)\calb(ty)\calb\left(-{\rm Ad}_{\calb(sx)^{-1}}sx\right)\\
&\stackrel{\eqref{eq:RRB0}}{=}& \dfrac{d^2}{dsdt}\Big|_{s,t=0} \calb(sx)\calb\left(ty-{\rm Ad}_{\calb(ty)\calb(sx)^{-1}}sx\right)\\
&\stackrel{\eqref{eq:RRB0}}{=}&
\dfrac{d^2}{dsdt}\Big|_{s,t=0} \calb\left(sx+{\rm Ad}_{\calb(sx)}ty -{\rm Ad}_{\calb(sx)\calb(ty)\calb(sx)^{-1}}sx\right)\\
&=&\calb|_{\ast 0}\left( \dfrac{d^2}{dsdt}\Big|_{s,t=0}
\left(sx+{\rm Ad}_{\calb(sx)}ty -{\rm Ad}_{\calb(sx)\calb(ty)\calb(sx)^{-1}}sx\right)
\right)\\
&=&\calb|_{\ast 0}\left( \dfrac{d^2}{dsdt}\Big|_{s,t=0}
{\rm Ad}_{\calb(sx)}ty -\dfrac{d^2}{dsdt}\Big|_{s,t=0}{\rm Ad}_{\calb(ty)}sx
\right)\\
&=&B([B(x),y]+[x,B(y)]).
\end{eqnarray*}
Thus, $B:\frakg\to \frak g$ is a Rota-Baxter operator of weight $0$ on $\frak g$.
\end{proof}

\begin{exam}\label{ex:adjoint_sl2}
Let $G={\rm SL}(2,\mathbb R)$ with its Lie algebra $\frakg={\rm sl}(2,\mathbb R)$. Then a smooth map $\calb:\frakg\to G$ is a Rota-Baxter operator of weight $0$ if and only if
$$\calb\begin{pmatrix}
a&b\\
c&-a
\end{pmatrix}\calb\begin{pmatrix}
x&y\\
z&-x
\end{pmatrix}=
\calb\left(\begin{pmatrix}
a&b\\
c&-a
\end{pmatrix}+\calb\begin{pmatrix}
a&b\\
c&-a
\end{pmatrix}\begin{pmatrix}
x&y\\
z&-x
\end{pmatrix}\calb\begin{pmatrix}
a&b\\
c&-a
\end{pmatrix}^{-1}\right)$$
for any $a,b,c,x,y,z\in\mathbb R$. For example, there are two special classes of solutions as follows:
$$\quad\calb_s\begin{pmatrix}
a&b\\
c&-a
\end{pmatrix}=\begin{pmatrix}
e^{sa}&0\\
0&e^{-sa}
\end{pmatrix}; \quad \quad\calb'_s\begin{pmatrix}
a&b\\
c&-a
\end{pmatrix}=\begin{pmatrix}
1&sc\\
0&1
\end{pmatrix}$$
for $s\in \mathbb R$. The differentiation of $\calb_s$ and that of $\calb'_s$ are respectively given by
$$B_s\begin{pmatrix}
a&b\\
c&-a
\end{pmatrix}=\begin{pmatrix}
sa&0\\
0&-sa
\end{pmatrix}\ \mbox{ and }\
B'_s\begin{pmatrix}
a&b\\
c&-a
\end{pmatrix}=\begin{pmatrix}
0&sc\\
0&0
\end{pmatrix}.$$
They are also two kinds of Rota-Baxter operators of weight $0$ on ${\rm sl}(2,\mathbb R)$ represented by $P_4$ and $P_5$ respectively in \cite[Theorem 2.1]{BGP}.

\end{exam}


\section{Relative Rota-Baxter operators of weight $0$ and pre-groups}\label{sec:pre}

In this section, we give the explicit relation between relative Rota-Baxter operators of weight 0 and pre-groups.

As introduced in \cite[Definition~2.5]{BGST}, a {\bf pre-group} is a triple $(G,+,\rhd)$, where $(G,+)$ is an abelian group and $\rhd:G\times G\to G$  is a multiplication on $G$ such that
\begin{eqnarray*}
\label{eq:pre-gp-1}
x\rhd(y+z)&=&(x\rhd y) + (x\rhd z),\\
\label{eq:pre-gp-1}
x\rhd(y\rhd z)&=&(x+x\rhd y)\rhd z
\end{eqnarray*}
for all $x,y,z\in G$.

Given a relative Rota-Baxter operator $\calr:V\to G$ on a group $G$ with respect to a $\mathbb Z G$-module $(V,\Phi)$,
one can define two products $\rhd_{\calr}$ and $*_{\calr}$ on $V$ respectively by
\begin{eqnarray}
\label{eq:action}
u \rhd_{\calr} v &=& \Phi(\calr(u))v,\\
\label{eq:descendent}
u *_{\calr} v &=& u+\Phi(\calr(u))v.
\end{eqnarray}

\begin{prop}\label{prop:pre-group}
The triple $(V,+,\rhd_{\calr})$ is a pre-group.
\end{prop}
\begin{proof} For any $u,v,w\in V$, we have
\begin{eqnarray*}
u\rhd_{\calr}(v+w) =\Phi(\calr(u))(v+w)
= \Phi(\calr(u))v+\Phi(\calr(u))w
= u\rhd_{\calr}v+u\rhd_{\calr}w,
\end{eqnarray*}
and
\begin{eqnarray*}
u\rhd_{\calr}(v\rhd_{\calr}w) &=& \Phi(\calr(u))(\Phi(\calr(v))w)\\
&=& \Phi(\calr(u)\calr(v))w\\
&\stackrel{\eqref{eq:RRB0}}{=}& \Phi(\calr(u+\Phi(\calr(u))v))w\\
&=& (u+u\rhd_{\calr}v)\rhd_{\calr}w\\
&=& (u*_{\calr}v)\rhd_{\calr}w.
\end{eqnarray*}
So $(V,+,\rhd_{\calr})$ is a pre-group.
\end{proof}

\begin{prop}\label{prop:rrb_pre-group}
The pair $(V,*_{\calr})$ is a group with unit $0$ and the inverse $u^{\dagger_{\calr}}$ of  $u\in V$ with respect to $*_{\calr}$ is given by
\begin{equation}\label{eq:des_inverse}
u^{\dagger_{\calr}}=-\Phi(\calr(u)^{-1})u.
\end{equation}
Hence, Eq.~\eqref{eq:RRB0} equivalently says that $\calr:(V,*_\calr)\to (G,\cdot)$
is a group homomorphism.

Define the left multiplication $L^\calr_u:V\to V$ for each $u\in V$ by
\begin{equation}\label{eq:left_mult}
L^\calr_u v=u\rhd_{\calr}v,\quad \forall v\in V.
\end{equation}
Then $L^\calr$ is an action of $(V,*_\calr)$ on $(V,+)$. In particular, we have the inverse formula
\begin{equation}\label{eq:left_mult_inverse}
(L^\calr_u)^{-1}=L^\calr_{u^{\dagger_{\calr}}},\quad \forall u\in V.
\end{equation}

We call $(V,*_{\calr})$ the {\bf descendent group} from the relative Rota-Baxter operator $\calr$.
\end{prop}
\begin{proof}
By \eqref{eq:RRB0}, we have
\begin{eqnarray*}
(u*_{\calr} v)*_{\calr} w
&=& u*_{\calr} v + \Phi(\calr(u*_{\calr} v))w\\
&=& (u+\Phi(\calr(u))v)+ \Phi(\calr(u)\calr(v))w\\
&=& u+\Phi(\calr(u))(v+\Phi(\calr(v))w)\\
&=& u*_{\calr} (v*_{\calr} w).
\end{eqnarray*}

Note that $\calr(0)=e_G$, so it is easy to see that $0*_{\calr} u=u*_{\calr} 0= u$.

On the other hand,
\begin{eqnarray*}
u*_{\calr}u^{\dagger_{\calr}}&=&u+\Phi(\calr(u))u^{\dagger_{\calr}}\\
&=& u+\Phi(\calr(u))(-\Phi(\calr(u)^{-1})u)\\
&=& u-\Phi(\calr(u)\calr(u)^{-1})u\\
&=& 0,\\[.5em]
u^{\dagger_{\calr}}*_{\calr}u&=&u^{\dagger_{\calr}}+\Phi(\calr(u^{\dagger_{\calr}}))u\\
&=&-\Phi(\calr(u)^{-1})u+\Phi(\calr(u)^{-1})\Phi(\calr(u))\Phi(\calr(u^{\dagger_{\calr}}))u\\
&=&\Phi(\calr(u)^{-1})(-u+\Phi(\calr(u))\Phi(\calr(u^{\dagger_{\calr}}))u)\\
&=&\Phi(\calr(u)^{-1})(-u+\Phi(\calr(u)\calr(u^{\dagger_{\calr}}))u)\\
&\stackrel{\eqref{eq:RRB0}}{=}& \Phi(\calr(u)^{-1})(-u+\Phi(\calr(u *_{\calr} u^{\dagger_{\calr}}))u)\\
&=& \Phi(\calr(u)^{-1})(-u+\Phi(\calr(0))u)\\
&=& \Phi(\calr(u)^{-1})(-u+\Phi(e_G)u)\\
&=& \Phi(\calr(u)^{-1})0\ \\
&=& 0,
\end{eqnarray*}
so $u^{\dagger_{\calr}}$ is the inverse of $u$ with respect to the product $*_{\calr}$.

Therefore, $(V,*_{\calr})$ is a group. All the other statements are straightforward.
\end{proof}

According to \cite[Theorem~4.3]{BGST}, any post-Lie group $(G,\rhd)$ with the smooth multiplication $\rhd$ gives a post-Lie algebra $(\frakg,\triangleright)$ by differentiation. Now any Rota-Baxter operator $\calb$ of weight $0$ on a Lie group $G$ especially provides a pre-Lie group $(\frakg,+,\rhd_\calb)$ by Proposition~\ref{prop:pre-group}. Consequently, we can also consider its differentiation. 
\begin{theorem}\label{thm:diff_pre_lie}
Let $\calb:\frak g\to G$ be a Rota-Baxter operator of weight $0$ on a Lie group $G$.
The differentiation of the pre-Lie group $(\frakg,+,\rhd_\calb)$ is the induced pre-Lie algebra $(\frakg,\triangleright_B)$ from the Rota-Baxter operator $B=\mathcal B|_{*0}$ of weight $0$ on $\frakg$ obtained in Theorem~\ref{coro:diff_rbo}. Namely, we have the following commutative diagram.
$$\xymatrix{(G,\calb) \ar@{->}[rr]^{\rm splitting} \ar@{->}[d]_{\rm differentiation}&& (\frakg,+,\rhd_\calb) \ar@{->}[d]^{\rm differentiation}\\
 (\frakg, B)\ar@{->}[rr]^{\rm splitting} && (\frakg, \triangleright_B)}$$
\end{theorem}
\begin{proof}
The differentiation of the smooth multiplication $\rhd_\calb$ on $\frakg$ is computed as follows:
\begin{eqnarray*}
\dfrac{d}{dt}\Big|_{t=0}\dfrac{d}{ds}\Big|_{s=0}tx\rhd_\calb sy
&\stackrel{\eqref{eq:rbo_lie},\,\eqref{eq:action}}{=}&\dfrac{d}{dt}\Big|_{t=0}\dfrac{d}{ds}\Big|_{s=0}{\rm Ad}_{\calb(tx)}sy\\
&=& \dfrac{d}{dt}\Big|_{t=0}{\rm Ad}_{\calb(tx)}y\\
&=& {\rm ad}_{\calb|_{*0}\left(\frac{d}{dt}|_{t=0}tx\right)}y\\
&=& [B(x),y]_\frakg\\
&=& x \triangleright_B y
\end{eqnarray*}
for any $x,y\in\frakg$. Hence, we get the induced pre-Lie algebra $(\frakg,\triangleright_B)$ from the differentiation $B$ of $\calb$.
\end{proof}

\section{Relative Rota-Baxter operators of weight $0$ and braces}\label{sec:bra}

In this section, we give explicit relation between relative Rota-Baxter operators of weight $0$ and braces.

Recall that a {\bf (left) brace} is a triple $(V,+,\ast)$, where $(V,+)$ is an abelian group and $(V,\ast)$ is a group such that
\begin{equation}\label{eq:brace}
x\ast(y+z)=(x\ast y) + (x\ast z)-x,\quad\forall x,y,z\in V.
\end{equation}
In particular, the zero element $0$ in $V$ is also the unit of $(V,\ast)$, and
$$x\ast(-y)=2x-x*y,\quad\forall x,y\in V.$$

\begin{theorem}\label{thm:brace}
Let $\calr$ be a relative Rota-Baxter operator of weight $0$ on $G$ with respect to the module $(V,\Phi)$. Then the triple $(V,+,*_\calr)$ is a brace, where $*_\calr$ is the multiplication defined in \eqref{eq:descendent}.
\end{theorem}
\begin{proof}
We only need to check Eq.~\eqref{eq:brace}:
\begin{eqnarray*}
x*_\calr(y+z) &=& x+\Phi(\calr(x))(y+z)\\
&=& x+\Phi(\calr(x))y+\Phi(\calr(x))z\\
&=& (x+\Phi(\calr(x))y)+(x+\Phi(\calr(x))z)-x\\
&=& (x*_\calr y) + (x*_\calr z)-x
\end{eqnarray*}
for any $x,y,z\in G$.
\end{proof}

\begin{exam}
  The derived brace $(\mathbb Z,+,*_t)$ from the relative Rota-Baxter operator given in Example \ref{ex:Z}  is a cyclic brace; see also \cite{Ru}.
\end{exam}

\begin{exam}
Consider the relative Rota-Baxter operator
  $\calr:V\to G,\ (a,x)\mapsto x$ given in Example \ref{ex:semi-direct}. Then     for any $a,b\in A$ and $x,y\in G$, we have
\begin{eqnarray*}
(a,x)*_\calr(b,y)&\stackrel{\eqref{eq:descendent}}{=}&(a,x)+\Phi(\calr(a,x))(b,y)=(a,x)+\Phi(x)(b,y)=(a+\rho(x)b,xy).
\end{eqnarray*}
Namely, the descendent group $(V,*_\calr)$ is just the semidirect product $A\rtimes_\rho G$; see also \cite[Example 1.4]{GV}.

Furthermore, take abelian groups $V=\mathbb Z_m\times\mathbb Z_n$ and $G=\mathbb Z_n$ for $m,n\in\mathbb Z^+$. Since any
$r\in \mathbb Z$ coprime to $m$ and $r^n\equiv 1\pmod m$ provides a group homomorphism $$\chi_r:\mathbb Z_n\to \mathbb Z_m^*,\ [i]\mapsto [r^i],$$
$V$ has the following $\mathbb Z G$-module defined by
$$[j]\cdot_r([k],[l])=([r^jk],[l]),\quad\forall j,k,l\in \mathbb Z.$$
Then the natural projection $t:\mathbb Z_m\times\mathbb Z_n\to \mathbb Z_n$ is a relative Rota-Baxter operator such that the descendent group $(V,*_t)$ is the semidirect product $\mathbb Z_m\rtimes_r\mathbb Z_n$

In particular, given two primes $p>q$ such that $p\equiv 1\pmod q$, there exists
$[r]\in \mathbb Z^*_p$ of multiplicative order $q$. The corresponding braces $(V,+,*_t)$ are a class of braces of cyclic type described in \cite[Theorem 3.4]{AB}.
Note that when $p=3$ and $q=2$, the descendent group $(V,*_t)=Z_3\rtimes_2\mathbb Z_2$ is also isomorphic to the symmetric group $S_3$.
\end{exam}

Relating to relative Rota-Baxter operators and braces, we also have the following notion.
\begin{defn}
Given a group $G$ and a $\mathbb Z G$-module $(V,\Phi)$, a map $\pi:G\to V$ is called a {\bf 1-cocycle}, if it satisfies
\begin{equation}\label{eq:1-cocycle}
\pi(xy)=\pi(x)+\Phi(x)\pi(y),\quad\forall x,y\in G.
\end{equation}
\end{defn}

Given a brace $(V,+,\ast)$, it is well-known that the identity map of $V$ is a bijective 1-cocycle with coefficients in the module $((V,+),\gamma)$ of $(V,\ast)$ defined by
$$\gamma(x)y=-x+x\ast y,\quad\forall x,y\in V,$$
where the induced group homomorphism
$\gamma:(V,*)\to \Aut(V,+)$
is called the {\bf gamma function} on $V$.
On the other hand, we have
\begin{prop}\label{prop:1-cocycle_rrb}
Given a group $G$ and a $\mathbb Z G$-module $(V,\Phi)$, if $\pi:G\to V$ is a
bijective $1$-cocycle with coefficients in $(V,\Phi)$, then $\pi^{-1}:V\to G$
is a relative Rota-Baxter operator.
\end{prop}
\begin{proof}
As $\pi:G\to V$ is bijective, the 1-cocycle condition \eqref{eq:1-cocycle}
is equivalent to say that
$$\pi^{-1}(x)\pi^{-1}(y)=\pi^{-1}(x+\Phi(\pi^{-1}(x))y).$$
So $\pi^{-1}:V\to G$
is a relative Rota-Baxter operator.
\end{proof}

In \cite{CS}, the authors discuss a natural question when
a skew brace can be induced by a Rota-Baxter operator of weight $1$ on a group.
A similar problem for Rota-Baxter operators of weight $0$ on Lie groups is considered as follows. Namely, given a Lie group $G$ and its Lie algebra $\frakg$,
when does a brace structure $(\frakg,+,*)$ on $\frakg$ come from a Rota-Baxter operator $\calb$ of weight $0$ on $G$? In other words, when does the coincidence $*=*_\calb$ happens?

First note that given a brace $(\frakg,+,*)$, if there is a Rota-Baxter operator $\calb$ of weight $0$ on $G$ such that
$*=*_\calb$, then the gamma function $\gamma$ on such a brace $(\frakg,+,*)$ is given by
\begin{eqnarray*}
\gamma(x) &\stackrel{\eqref{eq:rbo_lie},\,\eqref{eq:descendent}}{=}& {\rm Ad}_{\calb(x)},\quad \forall x\in\frak g.
\end{eqnarray*}
So $\gamma$ must have its image inside ${\rm Ad}_G$.
Now we are in the position to give our answer.
\begin{theorem}
Let $G$ be a connected Lie group with the associated Lie algebra $\frakg$. Let $(\frakg,+,*)$ be a brace such that its gamma function $\gamma$ takes values in ${\rm Ad}_G$, namely there is a function $C:\frakg \to G$ such that $\gamma(x) = {\rm Ad}_{C(x)}$ for all $x\in \frakg$. Define a map
$\kappa: \frakg \times \frakg \to Z(G)$
by
$$\kappa(x,y)=C(x)C(y)C(x*y)^{-1},\quad\forall x,y\in \frakg.$$
Then we have
\begin{enumerate}[{\upshape (i)}]
\item
$\kappa$ is a group $2$-cocycle of $(\frakg,*)$ with coefficients in the trivial $(\frakg,*)$-module $Z(G)$, whose cohomology class in
 $H^2((\frakg,*), Z(G))$ does not depend on the choice of $C$.

\item
The following statements are equivalent:

{\upshape (a)} The brace $(\frakg,+,*)$ comes from a Rota-Baxter operator of weight $0$ on $G$.

{\upshape (b)} The cohomology class of $\kappa$ in $H^2((\frakg,*), Z(G))$ is trivial.

\noindent
In this situation, $\kappa=\delta(\lambda)$ as a $2$-coboundary
for some $\lambda:\frakg \to Z(G)$, and the corresponding Rota-Baxter operator $\calb$ of weight $0$ on $G$ is given by
$$\calb(x)=C(x)\lambda(x)^{-1},\quad\forall x\in \frakg.$$

\item
Two Rota-Baxter operators $\calb, \calb'$ of weight $0$ on $G$ yield the same brace $(\frakg,+,*)$ if and only if there exists a group homomorphism $\xi:(\frakg,*)\to Z(G)$ such that $\calb(x)=\calb'(x)\xi(x)$
for all $x\in \frakg$.
\end{enumerate}
\end{theorem}
\begin{proof}
It is a classical fact that the kernel of the adjoint
representation ${\rm Ad}: G\to {\rm GL}(\frakg)$ coincides with the center $Z(G)$ when $G$ is a connected Lie group. On the other hand,
$${\rm Ad}_{C(x*y)}=\gamma(x*y)=\gamma(x)\gamma(y)={\rm Ad}_{C(x)}{\rm Ad}_{C(y)}={\rm Ad}_{C(x)C(y)},\quad\forall x,y\in \frakg,$$
so by definition
 $$\kappa(x,y)=C(x)C(y)C(x*y)^{-1}\in \ker{\rm Ad}=Z(G).$$

For (i), it is   easy to check that $\kappa$ satisfies the $2$-cocycle condition,
$$\kappa(y,z)\kappa(x*y,z)^{-1}\kappa(x,y*z)\kappa(x,y)^{-1}=e_G,\quad\forall x,y\in \frakg.$$
If there exists another function $C'$ such that $\gamma(x) = {\rm Ad}_{C'(x)}$ for all $x\in \frakg$, let
$$\kappa'(x,y)=C'(x)C'(y)C'(x*y)^{-1},\quad\forall x,y\in \frakg,$$
then one can see that $\kappa$ differs from $\kappa'$ by a $2$-coboundary
$\delta(\lambda)$, where
the function $\lambda:\frakg \to Z(G)$ is defined by $\lambda(x)=C(x)C'(x)^{-1}$ for all $x\in \frakg$.

For (ii), if $\calb$ is a Rota-Baxter operator of weight $0$ on $G$ such that
$*=*_\calb$, then $\kappa=\delta(\lambda)$ with $\lambda = C(\cdot)\calb(\cdot)^{-1}$ by taking $C'=\calb$ in the context above. Conversely, if $\kappa=\delta(\lambda)$ for some $\lambda:\frakg \to Z(G)$, let $\calb=C(\cdot)\lambda(\cdot)^{-1}$, then clearly ${\rm Ad}_{\calb(x)} = {\rm Ad}_{C(x)} = \gamma(x)$ and
$$\calb(x)\calb(y)=C(x)C(y)\lambda(y)^{-1}\lambda(x)^{-1}
=C(x*y)\kappa(x,y)\delta(\lambda)(x,y)^{-1}\lambda(x*y)^{-1}=\calb(x*y)$$
for all $x,y\in \frakg$, so $\calb$ is the desired Rota-Baxter operator.

For (iii), if $\calb$ and $\calb'$ are two Rota-Baxter operators of weight $0$ on $G$ such that $*_\calb=*_{\calb'}$, we equivalently have
$$\calb(x)\calb'(x)^{-1}\in \ker{\rm Ad}=Z(G),\quad\forall x\in \frakg,$$
then $\xi=\calb(\cdot)\calb'(\cdot)^{-1}$ is clearly a group homomorphism from $(\frakg,*)$ to $Z(G)$, where $*=*_\calb=*_{\calb'}$. Conversely, if $\calb(\cdot)=\calb'(\cdot)\xi(\cdot)$ for a group homomorphism $\xi:(\frakg,*)\to Z(G)$, then ${\rm Im}\,\xi\subset Z(G)$ implies that
$$x*_\calb y=x+{\rm Ad}_{\calb(x)}y=x+{\rm Ad}_{\calb'(x)\xi(x)}y
=x+{\rm Ad}_{\calb'(x)}y=x*_{\calb'}y,\quad\forall x,y\in \frakg,$$
so $*_\calb=*_{\calb'}$.
\end{proof}

\section{Relative Rota-Baxter operators of weight $0$ and the Yang-Baxter equation}\label{sec:YBE}
Since relative Rota-Baxter operators are deeply related to braces (Theorem~\ref{thm:brace}), they induce set-theoretic solutions
of the Yang-Baxter equation;~e.g. see the original papers \cite{Ru2, Ru0} of Rump and also \cite{CJO}. 
On the other hand, one can derive the same result from the aspect of pre-groups~\cite{BGST}. Both approaches are based on their intrinsic braided group structure~\cite{Lu,Ta}. 

\begin{defn}
Let $X$ be a non-empty set. Let $r\colon X^2\to X^2$ be
a map and write  $r(x,y)=(\sigma_x(y), \tau_y(x))$. We say that
$(X,r)$ is a {\bf non-degenerate involutive set-theoretic solution
of the Yang-Baxter equation} if
\begin{enumerate}[(i)]
\item $r^2=\id_{X^2}$,
\item $\sigma_x,\,\tau_x$ are permutations in the symmetric group $\Sym_X$ of $X$ for all $x\in X$,
\item $r_1r_2r_1=r_2r_1r_2$, where $r_1=r\times \id_X\colon X^3\longrightarrow X^3$
and $r_2=\id_X\times r\colon X^3\longrightarrow X^3$.
\end{enumerate}

\end{defn}

According to \cite[Proposition 2]{CJO}, Condition~(iii) for $(X,r)$ as a non-degenerate involutive set-theoretic solution of the Yang-Baxter equation can be equivalently expressed as follows,
\begin{enumerate}
\item[(iii)']
 $\sigma_x \sigma_{\sigma_x^{-1}(y)}=
\sigma_y \sigma_{\sigma_y^{-1}(x)},\quad\forall x,y\in X$.\label{it:ybs'}
\end{enumerate}
In this situation, we also have $\tau_y(x)=\sigma^{-1}_{\sigma_x(y)}(x)$.

Now we write down the formula for the induced set-theoretic solution
of the Yang-Baxter equation from a relative Rota-Baxter operator by modifying that from a brace given in \cite[Lemma~2~(iii)]{CJO}.
\begin{theorem}
Let $\calr$ be a relative Rota-Baxter operator of weight $0$ on $G$ with respect to the module $(V,\Phi)$, Then $\Upsilon_\calr:V^2\to V^2$ defined by
\begin{equation}\label{eq:YBsol}
\Upsilon_\calr(u,v)=\big(u\rhd_\calr v,(u\rhd_\calr v)^{\dagger_\calr}\rhd_\calr u\big),
\quad \forall u,v\in V,
\end{equation}
is a non-degenerate involutive set-theoretic solution
of the Yang-Baxter equation.
\end{theorem}



Applying Eqs.~\eqref{eq:action} and \eqref{eq:left_mult_inverse} to Eq.~\eqref{eq:YBsol}, we also obtain the following useful formula.
\begin{coro}
For any relative Rota-Baxter operator $\calr:V\to G$, its associated set-theoretic solution $\Upsilon_\calr$ of the Yang-Baxter equation can also be given by
\begin{equation}\label{eq:YBsol'}
\Upsilon_\calr(u,v)=\big(u\rhd_\calr v,\Phi(\calr(u\rhd_\calr v)^{-1}) u\big),
\quad \forall u,v\in V.
\end{equation}
\end{coro}

\medskip

Conversely, let $(X,r)$ be a non-degenerate involutive set-theoretic solution
of the Yang-Baxter equation.
The {\bf structure group} of $(X,r)$ is the
group $G(X,r)$ generated by  the elements of the set $X$ and
subject to the  relations $xy=zt$ for all $x,y,z,t\in X$ such that
$r(x,y)=(z,t)$.

Let $\mathbb{Z}^X$ denote the additive free abelian
group with basis $X$. Consider the natural action of $\Sym_X$ on
$\mathbb{Z}^X$ and the associated semidirect product
$\mathbb{Z}^X\rtimes \Sym_X$. According to \cite[Propositions 2.4 and
2.5]{ESS},
$$
H_r=\langle(x,\sigma_x)\,|\,x\in X\rangle
$$
is a subgroup of $\mathbb{Z}^X\rtimes \Sym_X$ isomorphic to $G(X,r)$, and there exists a function
\begin{equation}\label{eq:structure_group_rrb}
\calr_r:\mathbb{Z}^X\to \Sym_X
\end{equation}
such that $H_r$ is exactly the graph ${\rm Graph}(\calr_r)$ of $\calr_r$, namely $$H_r=\{(a,\calr_r(a))\,|\, a\in \mathbb{Z}^X\}.$$
In particular, $\calr_r(x)=\sigma_x$ for any $x\in X$.



Moreover, define a sum in $H_r$ by
$$(a,\calr_r(a))+(b,\calr_r(b))=(a+b,\calr_r(a+b)),\quad\forall a,b\in \mathbb{Z}^X.$$
Then \cite[Theorem 1]{CJO} tells us that $(H_r,+,\cdot)$ is a brace, called the {\bf underlying brace} of the solution $(X,r)$. Correspondingly, we obtain the following result.


\begin{theorem}\label{thm:syb-rrb}
Let $(X,r)$ be a non-degenerate involutive set-theoretic solution
of the Yang-Baxter equation.
\begin{enumerate}[{\upshape (i)}]
\item
The function $\calr_r$ in \eqref{eq:structure_group_rrb} is a relative Rota-Baxter operator on $\Sym_X$ with respect to its natural action on $\mathbb{Z}^X$.
\item
The projection map
$$\pi:H_r\to\mathbb{Z}^X,\ (a,\calr_r(a))\mapsto a$$
provides an isomorphism between the underlying brace $(H_r,+,\cdot)$ of  $(X,r)$ and $(\mathbb{Z}^X,+,*_{\calr_r})$ as in Theorem~\ref{thm:brace}. Moreover, the restriction of the associated solution $\Upsilon_{\calr_r}$ on $X^2$ coincides with $r$.
\end{enumerate}
\end{theorem}
\begin{proof}
For (i), it is due to Proposition~\ref{prop:rrb-graph}, as $H_r={\rm Graph}(\calr_r)$ is a subgroup of $\mathbb{Z}^X\rtimes \Sym_X$.

For (ii), since $H_r={\rm Graph}(\calr_r)$, it is clear that $\pi$ is bijective. Also,
\begin{eqnarray*}
\pi((a,\calr_r(a))+(b,\calr_r(b))) &=& \pi(a+b,\calr_r(a+b))\\
&=& a+b\ =\ \pi(a,\calr_r(a))+\pi(b,\calr_r(b)),\\[.5em]
\pi((a,\calr_r(a))(b,\calr_r(b)))&=& \pi(a+\calr_r(a)b,\calr_r(a)\calr_r(b))\\
&\stackrel{\eqref{eq:descendent},\,\eqref{eq:RRB0}}{=} &  \pi(a*_{\calr_r} b,\calr_r(a*_{\calr_r} b))\\
&=&  a*_{\calr_r} b\ =\  \pi(a,\calr_r(a))*_{\calr_r} \pi(b,\calr_r(b)).
\end{eqnarray*}
Hence, $\pi$ is an isomorphism of braces. Next for any $x,y\in X$, as $\calr_r(x)=\sigma_x$, we have
\begin{eqnarray*}
\Upsilon_{\calr_r}(x,y) \stackrel{\eqref{eq:YBsol'}}{=}  \big(x\rhd_{\calr_r} y,\calr_r(x\rhd_{\calr_r} y)^{-1}x\big)  \stackrel{\eqref{eq:action}}{=}
\big(\sigma_x(y),\sigma^{-1}_{\sigma_x(y)}(x)\big)= (\sigma_x(y),\tau_y(x)) \ =\ r(x,y),
\end{eqnarray*}
so we have checked that $\Upsilon_{\calr_r}|_{X^2}=r$.
\end{proof}

\begin{remark}
Theorem~\ref{thm:perm_repn} for permutation modules over symmetric groups actually reflects such a construction of relative Rota-Baxter operators.
\end{remark}

\begin{coro}
The inverse map of $\pi$ in Theorem~\ref{thm:syb-rrb} is a relative Rota-Baxter operator on $H_r$ with respect to its adjoint action on $\mathbb{Z}^X$ inside $\mathbb{Z}^X\rtimes \Sym_X$, namely
$${\rm Ad}_{(a,\calr_r(a))}b\ =\ \calr_r(a)b,\quad\forall a,b\in \mathbb{Z}^X.$$
\end{coro}
\begin{proof}
According to Theorem~\ref{thm:syb-rrb}, we have
\begin{eqnarray*}
\pi((a,\calr_r(a))(b,\calr_r(b))) &=& \pi(a,\calr_r(a))*_{\calr_r} \pi(b,\calr_r(b))\\
&\stackrel{\eqref{eq:descendent}}{=}& \pi(a,\calr_r(a))+\calr_r(a)\pi(b,\calr_r(b))\\
&=& \pi(a,\calr_r(a))+{\rm Ad}_{(a,\calr_r(a))}\pi(b,\calr_r(b)),
\end{eqnarray*}
so $\pi$ is a bijective $1$-cocycle of $H_r$ with coefficients in $(\mathbb{Z}^X,{\rm Ad})$. By Proposition~\ref{prop:1-cocycle_rrb}, we know that the inverse map $\pi^{-1}$
is a relative Rota-Baxter operator.
\end{proof}


\medskip
Next we apply formula \eqref{eq:YBsol'} to the aforementioned examples of relative Rota-Baxter operators to write down their induced set-theoretic solutions
of the Yang-Baxter equation.

\begin{exam}

Consider the vector representation $V=\mathbb R^2$ of the Lie group $G={\rm SL}(2,\mathbb R)$. Then a smooth map $\calr:V\to G$ is a relative Rota-Baxter operator if and only if
$$\calr\begin{pmatrix}
x\\
y
\end{pmatrix}\calr\begin{pmatrix}
z\\
w
\end{pmatrix}=
\calr\left(\begin{pmatrix}
x\\
y
\end{pmatrix}+\calr\begin{pmatrix}
x\\
y
\end{pmatrix}\cdot\begin{pmatrix}
z\\
w
\end{pmatrix}\right),\quad \forall x,y,z,w\in\mathbb R.$$

One simple solution is by setting
$$\calr\begin{pmatrix}
x\\
y
\end{pmatrix}=\begin{pmatrix}
1&f(x,y)\\
0&1
\end{pmatrix}$$
for some $f\in C^\infty(\mathbb R^2)$.
In this situation, we actually get a class of relative Rota-Baxter operators
$$\calr_s:V\to G,\ \begin{pmatrix}
x\\
y
\end{pmatrix}\mapsto \begin{pmatrix}
1&sy\\
0&1
\end{pmatrix}.$$
and its differentiation $R_s:V\to \frakg$ is given by $R_s\begin{pmatrix}
x\\
y
\end{pmatrix}=\begin{pmatrix}
0&sy\\
0&0
\end{pmatrix}$, where $\frakg={\rm sl}(2,\mathbb R)$.

We have
\begin{align*}
&\begin{pmatrix}
x\\
y
\end{pmatrix}\rhd_{\calr_s}
\begin{pmatrix}
z\\
w
\end{pmatrix}
=\begin{pmatrix}
1&sy\\
0&1
\end{pmatrix}\begin{pmatrix}
z\\
w
\end{pmatrix}
=\begin{pmatrix}
z+syw\\
w
\end{pmatrix},
\end{align*}
\begin{align*}
&\calr_s\left( \begin{pmatrix}
x\\
y
\end{pmatrix} \rhd_{\calr_s} \begin{pmatrix}
z\\
w
\end{pmatrix} \right)^{-1}\begin{pmatrix}
x\\
y
\end{pmatrix}
=\begin{pmatrix}
1&-sw\\
0&1
\end{pmatrix}\begin{pmatrix}
x\\
y
\end{pmatrix}
=\begin{pmatrix}
x-syw\\
y
\end{pmatrix}.
\end{align*}
So we have the following  set-theoretic solutions
of the Yang-Baxter equation on $\mathbb R^2$:
$$\Upsilon_{\calr_s}\left(\begin{pmatrix}
x\\
y
\end{pmatrix},
\begin{pmatrix}
z\\
w
\end{pmatrix}\right)
=
\left(\begin{pmatrix}
z+syw\\
w
\end{pmatrix},
\begin{pmatrix}
x-syw\\
y
\end{pmatrix}\right)
.$$
\end{exam}

\begin{exam}
In Example~\ref{ex:adjoint_sl2}, we give the following Rota-Baxter operators of weight $0$ on ${\rm SL}(2,\mathbb R)$,
$$\quad\calb_s\begin{pmatrix}
a&b\\
c&-a
\end{pmatrix}=\begin{pmatrix}
e^{sa}&0\\
0&e^{-sa}
\end{pmatrix}; \quad \quad\calb'_s\begin{pmatrix}
a&b\\
c&-a
\end{pmatrix}=\begin{pmatrix}
1&sc\\
0&1
\end{pmatrix}$$
for $s\in \mathbb R$,
and then
\begin{align*}
&\begin{pmatrix}
a&b\\
c&-a
\end{pmatrix}\rhd_{\calb_s}
\begin{pmatrix}
x&y\\
z&-x
\end{pmatrix}
=\begin{pmatrix}
e^{sa}&0\\
0&e^{-sa}
\end{pmatrix}
\begin{pmatrix}
x&y\\
z&-x
\end{pmatrix}
\begin{pmatrix}
e^{-sa}&0\\
0&e^{sa}
\end{pmatrix}
=\begin{pmatrix}
x&e^{2sa}y\\
e^{-2sa}z&-x
\end{pmatrix},\\
&{\rm Ad}_{\calb_s\left(\begin{pmatrix}
a&b\\
c&-a
\end{pmatrix}\rhd_{\calb_s}
\begin{pmatrix}
x&y\\
z&-x
\end{pmatrix}\right)^{-1}}
\begin{pmatrix}
a&b\\
c&-a
\end{pmatrix}=
{\rm Ad}_{\begin{pmatrix}
e^{-sx}&0\\
0&e^{sx}
\end{pmatrix}}
\begin{pmatrix}
a&b\\
c&-a
\end{pmatrix}
\\
&\qquad=
\begin{pmatrix}
e^{-sx}&0\\
0&e^{sx}
\end{pmatrix}
\begin{pmatrix}
a&b\\
c&-a
\end{pmatrix}
\begin{pmatrix}
e^{sx}&0\\
0&e^{-sx}
\end{pmatrix}=\begin{pmatrix}
a & be^{-2sx}\\
ce^{2sx} & -a
\end{pmatrix}.
\end{align*}
So we have the following  set-theoretic solutions
of the Yang-Baxter equation on ${\rm SL}(2,\mathbb R)$:
$$\Upsilon_{\calb_s}\left(\begin{pmatrix}
a&b\\
c&-a
\end{pmatrix},
\begin{pmatrix}
x&y\\
z&-x
\end{pmatrix}\right)
=
\left(\begin{pmatrix}
x & ye^{2sa}\\
ze^{-2sa} & -x
\end{pmatrix},
\begin{pmatrix}
a & be^{-2sx}\\
ce^{2sx} & -a
\end{pmatrix}\right)
.$$
Analogously, we have the following  set-theoretic solutions
of the Yang-Baxter equation on ${\rm SL}(2,\mathbb R)$:
$$\Upsilon_{\calb'_s}\left(\begin{pmatrix}
a&b\\
c&-a
\end{pmatrix},
\begin{pmatrix}
x&y\\
z&-x
\end{pmatrix}\right)
=
\left(\begin{pmatrix}
x+scz & y-2scx-s^2c^2z\\
z & -x-scz
\end{pmatrix},
\begin{pmatrix}
a-scz & b+2saz-s^2z^2c\\
c & -a+scz
\end{pmatrix}\right)
.$$
\end{exam}

\begin{exam}
For the associated solution $\Upsilon_{\calr_\sigma}$ of a relative Rota-Baxter operator $\calr_\sigma$ defined in
Theorem~\ref{thm:perm_repn}, its restriction on the basis $X_n=\{e_i\}_{1\leq i\leq n}$ is given as follows:
\begin{align*}
\Upsilon_{\calr_\sigma}(e_i,e_j)=\big(e_i\rhd_{\calr_\sigma} e_j,\Phi(\calr_\sigma(e_i\rhd_{\calr_\sigma} e_j)^{-1})e_i\big) =\big(e_{\sigma_i(j)},\Phi(\sigma^{-1}_{\sigma_i(j)})e_i\big) =\big(e_{\sigma_i(j)},e_{\sigma^{-1}_{\sigma_i(j)}(i)}\big).
\end{align*}
\end{exam}

\begin{exam}
In Example~\ref{ex:semi-direct}, we give a kind of relative Rota-Baxter operators
 $$\calr:V\to G,\ (a,x)\mapsto x,$$
where $V=A\times G$ for an abelian group $(G,\cdot)$, a $\mathbb Z G$-module $(A,+,\rho)$ and the group action $\Phi:G\to\Aut(V)$ defined by $\Phi(x)(a,y)=(\rho(x)a,y)$. Then we have
\begin{align*}
 (a,x)\rhd_\calr(b,y)=\Phi(x)(b,y)=(\rho(x)b,y),
\end{align*}
and
\begin{align*}
 \Phi(\calr((a,x)\rhd_\calr(b,y))^{-1})(a,x)=
\Phi(\calr(\rho(x)b,y)^{-1})(a,x) =\Phi(y^{-1})(a,x)=(\rho(y)^{-1}a,x).
\end{align*}
So we have the following  set-theoretic solutions
of the Yang-Baxter equation on $V$:
$$\Upsilon_\calr((a,x),(b,y))=
((\rho(x)b,y),(\rho(y)^{-1}a,x)).
$$
\end{exam}

\section{Relative Rota-Baxter operators of weight $0$ and $T$-structures}\label{sec:T}
In \cite[Appendix]{ESS}, the authors introduced the notion of $T$-structure on abelian groups deeply related to bijective 1-cocycles.  A permutation map $T$ of an abelian group $A$ is called a {\bf $T$-structure} on $A$ if
\begin{equation}\label{eq:T-structure}
T(ka)=kT^k(a),\quad\forall a\in A,\,k\in\mathbb Z.
\end{equation}

Any bijective 1-cocycle $\pi:G\to A$ with coefficients in a $\mathbb Z G$-module $(A,\rho)$ defines a $T$-structure $T$ on $A$ by
\begin{equation}\label{eq:cocycle-T-struc}
T(a)=\rho(\pi^{-1}(a)^{-1})a,\quad\forall a\in A.\end{equation}

On the other hand, according to \cite[Theorem A.7]{ESS}, any $T$-structure $T$ on a cyclic group $A=\langle 1\rangle$ completely determines a bijective cocycle datum $(G,A,\rho,\pi)$ as follow. Let $G=A$ as a set and $\pi:G\to A$ be the identity map. Interpret $A$ as a quotient ring of $\mathbb Z$ and define the map $\rho:G\to\Aut(A)$ by
$$\rho(m)n=nT^{-T(m)}(1),\quad\forall m,n\in \mathbb Z.$$
Also, define the multiplication $\ast$ on $G$ by
$$m\ast n=m+\rho(m)n=m+nT^{-T(m)}(1),\quad\forall m,n\in \mathbb Z.$$
Then $(A,\rho)$ is a module of $(G,\ast)$, and $\pi:G\to A$ is a bijective 1-cocycle with coefficients in $(A,\rho)$.
 Note that the multiplication $\ast$ of $G$ defined here is the opposite of $\circledcirc $ given in the proof of \cite[Theorem A.7]{ESS}, since our 1-cocycle condition \eqref{eq:1-cocycle} is slightly different from there.  The classification of $T$-structures on cyclic groups, or equivalently cyclic braces, has been completed in \cite{Ru}.

The next result gives the relationship between $T$-structures and relative Rota-Baxter operators.
\begin{theorem}
Let $\calr$ be a relative Rota-Baxter operator of weight $0$ on $G$ with respect to the module $(V,\Phi)$. Then the map $T_\calr:V\to V$ defined by
\begin{equation}\label{eq:RRBO-T-map}
T_\calr(v)=-v^{\dagger_{\calr}}\stackrel{ \eqref{eq:des_inverse}}{=}\Phi(\calr(v)^{-1})v,\quad\forall v\in V,
\end{equation}
is a $T$-structure on $V$.
\end{theorem}
\begin{proof}
First by definition we see that the so-defined map $T_\calr:V\to V$ is a bijection, and we abbreviate it as $T$ in the proof.

Next we prove the following formula
\begin{equation}\label{eq:RRBO-T-struc}
\Phi(\calr(nv)^{-1})mv=mT_\calr^n(v),\quad\forall v\in V,\ m,n\in \mathbb N,
\end{equation} by induction on $n$.
The case when $n=0$ is trivial. Now suppose that it holds when $n\leq k$.
Taking inverses of the both hand sides of Eq.~\eqref{eq:RRB0} and then applying $\Phi$, we get
$$\Phi(\calr(v)^{-1})\Phi(\calr(u)^{-1})=\Phi(\calr(u+\Phi(\calr(u))v)^{-1}).$$
Setting $v=kT(u)$ and applying to $mu$, we have
$$\Phi(\calr(kT(u))^{-1})\Phi(\calr(u)^{-1})mu=\Phi(\calr(u+\Phi(\calr(u))kT(u))^{-1})mu.$$
By the induction hypothesis, we see that the LHS is
\begin{eqnarray*}
\Phi(\calr(kT(u))^{-1})\Phi(\calr(u)^{-1})mu &=& \Phi(\calr(kT(u))^{-1})(m\Phi(\calr(u)^{-1})u)\\
&=& \Phi(\calr(kT(u))^{-1})(mT(u))\\
&=& mT^k(T(u)) = mT^{k+1}(u),
\end{eqnarray*}
while the RHS is
\begin{eqnarray*}
\Phi(\calr(u+\Phi(\calr(u))kT(u))^{-1})mu &=& \Phi(\calr(u+k\Phi(\calr(u))T(u))^{-1})mu\\
&=&  \Phi\left(\calr\left(u+k\Phi(\calr(u))\Phi(\calr(u)^{-1})u\right)^{-1}\right)mu\\
&=& \Phi\left(\calr((k+1)u)^{-1}\right)mu.
\end{eqnarray*}
Letting $m=n$ in Eq.~\eqref{eq:RRBO-T-struc}, that is exactly Eq.~\eqref{eq:T-structure}, so $T_\calr$ is a $T$-structure on $V$.
\end{proof}
\begin{remark}
For a relative Rota-Baxter operator $\calr$ on $G$ with respect to the module $(V,\Phi)$, we have known that $\id_V:(V,*_\calr)\to (V,+)$ is a bijective 1-cocycle with coefficients in $(V,\rhd_\calr)$. Hence,
we can also see that
$$T_\calr(v)
\stackrel{\eqref{eq:cocycle-T-struc}}{=}
v^{\dagger_\calr}\rhd_\calr v
\ \left(\stackrel{\eqref{eq:left_mult_inverse}}{=}
(L^\calr_v)^{-1}v\right)
=-v^{\dagger_\calr}$$
defines a $T$-structure $T_\calr$ on $V$ by this way.
\end{remark}

\vspace{0.1cm}
 \noindent
{\bf Acknowledgements. } This research is supported by NSFC (Grant Nos.   12071094, 12371029) and Guangdong Basic and Applied Basic Research Foundation (Grant No. 2022A1515010357).

\bibliographystyle{amsplain}

\end{document}